%% file: paper.tex
\documentclass[runningheads,a4paper,envcountsame,11pt]{llncs}

\pdfoutput=1
\usepackage{comment}
\usepackage{booktabs} 
\usepackage{xspace, authblk}

\usepackage[ruled]{algorithm2e} 

\SetAlFnt{\small}
\SetAlCapFnt{\small}
\SetAlCapNameFnt{\small}
\SetAlCapHSkip{0pt}
\IncMargin{-\parindent}

\usepackage{latexsym}
\usepackage{epsfig}
\usepackage{verbatim}
\usepackage{shadow}
\usepackage{amssymb}
\usepackage{amsmath}
\usepackage{mathtools}

\usepackage{amsthm}
\usepackage{graphicx}
\usepackage{color}
\usepackage{bm}
\usepackage[authoryear,round]{natbib}
\renewcommand{\cite}{\citep} 
\usepackage{url}

\usepackage{thmtools, thm-restate}
\usepackage{subcaption}
\usepackage{physics}

\newcommand{\eps}{\varepsilon}
\newcommand{\CN}{{\mathcal N}}
\newcommand{\CF}{{\mathcal F}}
\newcommand{\ff}{\mbox{\boldmath $f$}}
\newcommand{\tff}{\mbox{\boldmath $\tilde{f}$}}

\newcommand{\prt}{\partial}

\def\NE{Nash equilibrium\xspace}

\def\NEa{Nash equilibria\xspace}

\def\nr{\mathrm{\textit{NR}}}

\newcommand{\positiveintegers}{\mathbb{Z}^+}

\newcommand{\players}{\mathcal{P}}

\newcommand{\graph}{\mathcal{G}}
\newcommand*{\game}{(\players, \graph, \bm{\beta}, \bm{u})}

\usepackage{cleveref}
\hypersetup{
	colorlinks=true,
	linkcolor=blue,
	citecolor=blue,
}

\begin{document}
\title{Resource Allocation Game on Social Networks: Best Response Dynamics and Convergence}
\author{
	Wei-Chun Lee\inst{1}
	\and
	Vasilis Livanos\inst{2}
	\and
	Ruta Mehta\inst{2}
	\and
	Hari Sundaram\inst{2}
}

\institute{Google LLC, 
		\email{wclee@google.com}
\and
		Department of Computer Science, 
		University of Illinois at Urbana-Champaign, \\
		\email{livanos3,rutameht,hs1@illinois.edu}
		}

\maketitle

\input{abstract.tex}
\input{introduction.tex}
\input{preliminaries.tex}
\input{convergence.tex}
\input{game-properties.tex}
\input{experiments.tex}
\input{conclusion.tex}

\bibliographystyle{abbrvnat}
\bibliography{rutarefs.bib,harirefs.bib}

\newpage

\input{appendix.tex}

\end{document}

%% file: abstract.tex

\begin{abstract}
The decisions that human beings make to allocate time has significant bearing on economic output and to the sustenance of social networks. The time allocation problem motivates our formal analysis of the resource allocation game, where agents on a social network, who have asymmetric, private interaction preferences, make decisions on how to allocate time, a bounded endowment, over their neighbors. Unlike the well-known opinion formation game on a social network, our game appears not to be a potential game, and the {\em Best-Response dynamics} is non-differentiable making the analysis of Best-Response dynamics non-trivial.

In our game, we consider two types of player behavior, namely {\em optimistic} or {\em pessimistic}, based on how they use their time endowment over their neighbors. 
To analyze Best-Response dynamics, we circumvent the problem of the game not being a potential game, through the lens of a novel two-level potential function approach. We show that the Best-Response dynamics converges point-wise to a Nash Equilibrium when players are all: {\em optimistic}; {\em pessimistic}; a mix of both types. 
Finally, we show that the Nash Equilibrium set is non-convex but connected, and Price of Anarchy is unbounded while Price of Stability is one. Extensive simulations over a stylized grid reveals that the distribution of quality of the convergence points is unimodal---we conjecture that presence of unimodality is tied to the connectedness of Nash Equilibrium.

\end{abstract}

%% file: introduction.tex

\section{Introduction}
\label{sec:Introduction}

Social networks play a central role in our lives, both for personal and professional growth. Information diffusion over a social network is one extensively analyzed class of problems~\citep{FILW14,Bakshy2012,Myers2012,Chierichetti2011,Karp2000}. Within the Theoretical CS community, the information diffusion problem is often modeled as a {\em network game} where a player's utility depends on her own and her neighbor's actions, {\em e.g.,} \citep{CKO13,FILW14,FGV16}. A common underlying assumption: a player has adequate resources (e.g. time) to interact with each of her neighbors, and that her payoff is independent of the amount of time spent with each of its neighbors. In the real world, agents have a finite endowment of resources. The dynamics of network games where agents have bounded resources is less well studied and is the focus of this paper.

Research on how individuals spend time, has a rich history in Economics~\citep{Juster1991,Becker1965,Gronau1977}. As~\citet{Juster1991} state in the introduction to their article, ``... the fundamental scarce resource in the economy is the availability of human time, and that the allocation of time to various activities will ultimately determine the relative prices of goods and services, the growth path of real output, and the distribution of income.'' The work of~\citet{Becker1965} spurred the analysis of ``non-market'' time (i.e. outside of work), in particular how individuals utilized their time at home and consequences of time use at home on the market.~\citet{Gronau1977} further sharpened the  analysis of ``non-market time'' by distinguishing between time spent on ``work at home'' and ``leisure.''~\citet{Juster1991} further identified ``socializing'' as a component of ``leisure.''

Despite having a limited amount of time (an inelastic resource) each day, we spend it on sustaining and expanding our social capital. While there is no formal definition of ``social capital'', there is a growing consensus that ``social capital stands for the ability of actors to secure benefits by virtue of membership in social networks or other social structures.''~\cite{Portes1998}. To develop social capital,~\citet{Miritello2013a} suggests that we spend time on communication to maintain friendships and~\cite{Roberts2011} finds that expending time is crucial to the sustenance of social networks.
This connection between expenditure of time and its effect on the social network motivates the following question:

\begin{center}
	{\em How would social interactions evolve on a social network when individuals are resource constrained?}
\end{center}

Consider for example, an academic who receives requests on her time to meet for the coming week: from her PhD students, undergrads who are taking her class, colleagues who want to go to lunch, the occasional meeting request from her department chair. The academic may have private preferences with whom she should spend time and will respond with counter proposals: a PhD student who wants to meet for an hour may get thirty minutes; an undergrad fifteen minutes; agree to spend an hour with her peers; agree to meet for the time that her chair asks. Each person asking her for time, is simultaneously responding to requests on their time from \textit{their} social network. In each case, when someone asks for time, the participants will typically agree to meet for the smaller of the two proposed times. As individuals in a social network make decisions on allocating time, we would like to know: does the social network converge to an equilibrium? If it does, is it stable? In other words, how sensitive is the network to a small change in proposals by any one person? What is the social welfare at equilibria?

Motivated by these scenarios, we define a {\em resource allocation} game on a social network, and analyze the game for convergence of the {\em best-response} dynamics to Nash equilibria.
Next, we summarize our technical contributions.

\subsection{Our Technical Contributions}
\label{sub:Technical Contributions}

We define the {\em resource allocation} game on a (social) network $G=(V,E)$ where each node $i\in V$ is a rational player (agent) with finite time endowment, say $\beta_i$. A player $i$ obtains utility by spending time with her network neighbors $\mathcal{N}_i$ which may differ from neighbor to neighbor. First, she has a {\em private} interaction preferences over her neighbors, namely $w_{ij}$ for each $j\in \CN_i$ such that $\sum_{j \in \CN_i} w_{ij}=1$. Second, to capture decreasing marginal returns, a non-negative, concave, and increasing function $u_{ij}$ captures $i$'s utility from interaction with $j \in \CN_i$. 
In general, we consider the asymmetric case with respect to (w.r.t) the interaction preferences $w_{ij}$, i.e., $w_{ij} \neq w_{ji}$, and we consider both symmetric and asymmetric cases w.r.t the interaction utility $u_{ij}$. We consider two types of player behavior, namely {\em optimistic} and {\em pessimistic}, based on their aggressiveness in proposing time to neighbors.


A pair of players $(i, j)$ who are neighbors i.e. $(i,j)\in E$, can interact only when both agree to do so. Thus, if $f_{ij}$ is the (time) interaction frequency ``proposed'' by $i$ to $j$, and $f_{ji}$ by $j$ to $i$, then the pair will agree to $f^*_{i,j} = \min\{f_{ij},f_{ji}\}$. The total utility of agent $i$ is:
\[
	\sum_{j\in \mathcal{N}_i} w_{ij} u_{ij}(f^*_{ij}),
\]



The main contribution of this paper is the analysis of {\em best-response} (BR) dynamics of the resource allocation game. In each round, a chosen player (in some arbitrary sequence) plays her utility maximizing (best-response) strategy given the proposals of her network neighbors. This game exhibits non-continuous dynamics, and is non-trivial to analyze. Except for weighted potential games \citep{monderer:96}, BR dynamics do not converge for most games. 

Most known results on best response dynamics in (social) network games are under following settings, symmetry of weights along edges, discrete choices, and no $u_{ij}$'s, {\em e.g.,} \cite{FGV16,FV17,FILW14}. 
The symmetric games are known to be potential games~\citep{FGV16} and thus the BR dynamics converge to Nash. The general form of our game exhibits asymmetry in both $w_{ij}$ and $u_{ij}$, and thereby our game does not seem to be a (weighted) potential game. 
Despite this, we obtain convergence for a special case and the general case, described next. These results are independent of the player behavior; optimistic, pessimistic, or a mix of two. 

We first consider the case of {\em global ranking weight system} (Section~\ref{sec:conv:seq:gr}): there is an intrinsic social order ({\em e.g.,} an academic hierarchy, comprising say professors, Ph.D. students and undergrads, continuing our stylized example) among players, and the $w_{ij}$'s are proportional to the global social rank, {\em i.e.,} if $R(j)$ is the rank of player $j\in V$ then $w_{ij}=\frac{R(j)}{\sum_{k\in \CN_i} R(k)}$. Note that, still $w_{ij}\neq w_{ji}$. Furthermore, we assume that the utility functions are symmetric on every edge --- $u_{ij} = u_{ji}$ for all $(i,j) \in E$. 

\begin{theorem}[Informal]
	The resource allocation game with global ranking weight system is a weighted potential game. As a consequence, the best-response dynamics in these games converges to a Nash equilibrium.
\end{theorem}

Next we consider the general resource allocation game when both weights $w_{ij}$ and utilities $u_{ij}$ are asymmetric and show the convergence of BR dynamics (in Section~\ref{sec:conv:seq:gen}), despite its seemingly non-potential nature.

\begin{theorem}[Informal]
	Best response dynamics converges to a Nash equilibrium in general resource allocation games with asymmetric preferences.
\end{theorem}

Our method of proof is akin to the construction of a ``two-level'' potential-like function using the unused total time of a specific class of players. We show that the outer function is monotonically decreasing, and when the outer function is fixed, the inner function decreases. We note that, when the outer function is decreasing, there are no guarantees for the inner function.

What if players move {\em simultaneously}? For the {\em simultaneous-play} variant of the best-response dynamics in our game, we show the existence of a cycle through a stylized example. Note that potential games \citep{monderer:96} {\em e.g.,} battle of the sexes may also exhibit cyclic behavior under the simultaneous-play.


We show that, the Price-of-Anarchy (PoA) is unbounded while Price-of-Stability (PoS) is one (Section \ref{sec:proper:pos-poa}). Interestingly, we can show that the set of Nash Equilibrium is convex under certain conditions (when all players follow a particular strategy) and is connected in general.

Experiments on a stylized grid show that best response dynamics converge to equilibria with high social-welfare (see Section \ref{sec:exper}). 
Experiments also demonstrate that the social-welfare distribution of points where BR dynamics converge is unimodal---probably a consequence of connectedness of the equilibrium set. We leave open a formal analysis of this unimodality,  
and the qualitative analysis of the social welfare of convergence points through notion of {\em average-price-of-anarchy}~\cite{PP16}.

\input{related.tex}

To summarize: our main technical contribution lies in the analysis of best response dynamics of the resource allocation game where agents have bounded endowment and private, asymmetric interaction preferences. The main challenge: our game is in general not a potential game, but we are able to show, through a novel two-level potential-like function approach, convergence to Nash Equilibria. We analyze PoA, PoS, and characterize the quality 
of Nash Equilibria.

The rest of this paper is organized as follows. In the next section, we formally introduce the game model. 
Then, in~\Cref{sec:conv}, we present convergence results for best-response dynamics. 
In~\Cref{sec:proper}, we prove guarantees for several key properties, including the Price of Anarchy and the Price of Stability. We conclude by presenting experimental results showing the distribution of the quality of Nash Equilibrium in~\Cref{sec:exper}, and summarize and discuss future directions in~\Cref{sec:Conclusion}.

%% file: related.tex

\subsection{Related Work}
\label{sub:Prior Work}

The work most related to ours is of \citet{Anshelevich2012} studying {\em contribution games on networks} with symmetric utility functions on each edge. This translates to unweighted case with $u_{ij}=u_{ji}$ for each edge $(i,j) \in E$ in our model. They analyze 2-strong equilibria, where no pair of two players can deviate and both gain, for efficiency (PoA), and shows convergence of best response dynamics. While the notion of 2-strong equilibria is stronger than the notion of a \NE, we observe that the complexity of analyzing best response dynamics in our proposed work comes from non-symmetric weights on edges.

There has been extensive work on information diffusion in social networks, where decisions/opinions are discrete, typically binary, and graphs are weighted (but no $u_{ij}$ functions). For example, \cite{FGV16} study a {\em discrete preference game}, and show that it is a potential game and therefore the best-response dynamics converges to a Nash equilibrium. In addition they show a polynomial time convergence rate for unweighted graphs, and pseudo-polynomial time convergence rate for the weighted graphs. \cite{FV17} studies this game under {\em social pressure} and obtain fast convergence in special cases. \cite{FILW14} studies the consensus game \citep{DeGroot1974} under asynchronous updates and the majority rule -- a special case of the linear threshold model. \cite{CKO13} studies the price-of-stability of the game when the edge functions interpolates between symmetric (coordination) and non-symmetric (unilateral decision-making). Furthermore, a number of works have explored learning \citep{Bala1998a,Acemoglu2011,NPS15} and herd behavior \citep{Banerjee1992}. 

Our work focuses on best-response dynamics. The best-response dynamics does not converge generally due to its discontinuous nature, while in coordination, or more generally in congestion games, the dynamics converges to a pure Nash equilibrium~\cite{MontanariSaberi}. Apart from BR dynamics, there is extensive literature on the analysis of no-regret dynamics~\cite{no-regret-book1,shai2012} within algorithmic game theory, see~\cite{no-regret-ch}. For general games, the average of the points visited by dynamics converges, but to correlated equilibria, a weaker notion than Nash equilibrium~\cite{blum-no-regret}. While in case of coordination games, the dynamics converges point-wise to a pure Nash equilibrium~\cite{LosertAkin,MPP}. ~\citet{PanageasPiliouras} show that, while the social welfare at the limit-points of the dynamics may not be near optimal, the expected welfare (average price-of-anarchy) is almost optimal in a few special cases. Fast convergence of the average is known for various special cases, e.g., see \cite{fast1,fast2,fast3}.


Work in Ecological games on foraging and more generally on predator-prey dynamics~\cite{Brown1999,Rosenzweig1963} analyzes the outcomes of agents who make decisions with resource constraints. The broad idea is that species forage for food, under limited energy constraints. However, most works focus on population dynamics (growth and depletion of species) which is not a focus of our paper.


%% file: preliminaries.tex

\section{Preliminaries}
\label{sec:prelim}

In this section we first formalize the game played on a social network by the resource constrained agents, its dynamics, and its Nash equilibria under two different types of player behavior. 

\subsection{Game Model}
\label{sec:prelim:model}

Consider a social network with $n$ agents (players) represented by an undirected graph $G=(V,E)$, where $V$ represents the players and $E$ the links between them. The players are numbered $1$ through $n$, and we denote the set of all players as $\players$. An edge (link) $(i,j)\in E$ represents the interaction between player $i$ and player $j$. The set of neighbors of player $i\in V$ is denoted by $\CN_i=\{ j \ |\ (i,j)\in E\}$ -- note that $i$ can be in $\CN_i$.

Players gain utility from communicating/interacting with each other, however the amount of resources to communicate, such as time, is available in limited quantity to each player. In particular, player $i$ has $\beta_i \ge 0$ amount of communication resource that she can distribute among her neighbors. Let $f_{ij} \ge 0$ denote the frequency proposal made by player $i$ to her neighbor $j$. It follows that $\sum_{j \in \CN_i} f_{ij} \le \beta_i.$ If we denote the vector of frequency proposals of $i$ to all players in $\CN_i$ by $\ff_i= (f_{i1}, \dots, f_{in})$, the set of strategies (all possible allocations) of player $i$ is

\[
\CF_i = \left\{ \ff_i \ \bigg| \ \ff_i \ge 0,\ \ \sum_{j \in \CN_i} f_{ij} \le \beta_i \right\}
\]

where it is understood that $f_{ij} = 0$ if $j \notin \CN_i$.
For communication to happen between players $i$ and $j$, naturally both have to agree to do so. Therefore, the realized allocation of resource, also called {\em interaction frequency} from now on, on edge $(i,j)$ is

\[
f^*_{ij} = \min\{f_{ij},f_{ji}\}.
\]

In other words, $f^*_{ij}$ denotes the {\em agreed upon interaction frequency} between $i$ and $j$. To differentiate we will call $f_{ij}$, {\em proposed interaction frequency}.

To capture asymmetric liking of a player and her neighbor, we consider a weighted network with asymmetric weights. The weight assigned by player $i$ to her neighbor $j$ is denoted by $w_{ij}$. Note here that $w_{ij}$ and $w_{ji}$ may be different. Once every player decides her allocation, let $\ff=(\ff_1,\dots,\ff_n)$ denote the allocation profile of all the players, and let $\ff^*$ denote the agreed upon allocation. The utility of player $i$ at profile $\ff$ is

\begin{equation}\label{eq:util}
u_i(\ff) = \sum_{\mathclap{\substack{j \in \CN_i \\ j \neq i}}} {w_{ij} u_{ij}(f^*_{ij}) }
\end{equation}

where $u_{ij}(f^*_{ij})$ is a non-negative increasing concave function of $f^*_{ij}$ and therefore captures decreasing marginal returns.
We are now ready to formally define our game.

\begin{definition}[Game]\label{def:game}
A {\em game} consists of a weighted graph $\graph=(V,E)$ where the nodes are the players $\players=V$, links $E$ represent the underlying social structure, and weights $w_{ij}$ and $w_{ji}$ on link $(i,j)\in E$ represent player preferences.
For each link $(i,j)\in E$ we are given functions $u_{ij}$ and $u_{ji}$ capturing respectively utility of player $i$ from interaction with $j$ and vice-versa.
Vector $\bm{\beta}$ represents amount of resources of all the players, where the resource constraint of player $i \in \players$ is $\beta_i$. We denote such a game by $\game$.
\end{definition}

Interestingly, our convergence results hold for any arbitrary non-negative, increasing and concave function $u_{ij}$ of $f^*_{ij}$; again
$u_{ij}$ and $u_{ji}$ need not be the same. This leads us to the definition of the social welfare of our game.

\begin{definition}[Social Welfare]\label{def:sw}
Let $\game$ be a game and $\ff$ a frequency profile of this game. Then, the {\em social welfare} of $\game$ at $\ff$ is

\[
SW(\ff) = \sum_{i \in \players} {u_i(\ff)}
\]
\end{definition}

\subsubsection{The Global Ranking Model.}
\label{sec:prelim:gr-model}

\par Many a times there is an inherent hierarchy among the players in a social network, {\em e.g.} the social network of a
company, a network of tennis players, etc. Taking this as motivation, we define a special case of our model, in which there
exists a global ranking of players capturing their social status within the network, and the weights $w_{ij}$ reflect this
global ranking.

\begin{definition}[Global-Ranking Weight System]\label{def:ranking}
 A global ranking is a function $R:\players\mapsto\positiveintegers$. By imposing this function on a game $\game$, we can associate
 each player $i$ with a number $R(i)$. The corresponding \textit{global ranking weight system} is a weighting scheme in
 which the weight $w_{ij}$ that player $i$ places on player $j$ is defined as

 \begin{equation}\label{eq:global-ranking}
  w_{ij} = \frac{R(j)}{\nr(i)} \ \ \ \ \ \ \mbox{where } \nr(i)=\sum_{k \in \CN_i} {R(k)}
 \end{equation}
\end{definition}

\subsection{Optimistic/Pessimistic Agents and Nash Equilibria}
\label{sec:prelim:opt-pes}

In this section we discuss the Nash equilibria of our game under two types of players, namely optimistic and pessimistic.
For a given profile $\ff$, to denote proposals of all player but $i$'s we use $\ff_{-i}$.

\begin{definition}[Nash Equilibrium]\label{def:NE}
A strategy profile $\ff$ is said to be a {\em Nash equilibrium} if no player gains utility by unilateral deviation \cite{Nash1951}, i.e.,
\[
\forall i \in V,\ \ \ u_i(\ff) \ge u_i(\ff'_i,\ff_{-i}), \ \ \ \forall \ff'_i \in \CF_i
\]
\end{definition}

At Nash equilibrium every player is playing her utility maximizing strategy given everyone else's strategy.
Given everyone's frequency proposals, the utility maximizing proposal of a player is called her {\em best-response}. Thus, at NE every player is playing a best-response to the strategies of the other players. 
The best-response of player $i$ with respect to a proposal profile $\ff$ can be computed using the following convex program.

\begin{equation} \label{eq:BR}
\begin{array}{lll}
max  & \qquad \displaystyle{\sum_{\mathclap{\substack{j \in \CN_i \\ j \neq i}}} {w_{ij} u_{ij}(f^*_{ij}) } } & \\
s.t. & \qquad f^*_{ij} \leq f_{ij} \quad \& \quad f^*_{ij} \le f_{ji} \qquad & \forall j \in \CN_i \\
     & \qquad \sum_{j \in \CN_i} f_{ij} \leq \beta_i & \\
     & \qquad f^*_{ij} \geq 0 & \forall j \in \CN_i \\
     & \qquad f_{ij} \geq 0 & \forall j \in \CN_i \\
\end{array}
\end{equation}

where note that $f_{ji}$, the proposal of $j$ to $i$ is a constant. 
For players $i$ and $j$ we have $f^*_{ij} = \min{ \{ f_{ij}, f_{ji} \} }$, if the
proposals are not exactly equal, then  one of the players will have some ``leftover'' frequency because they had to settle for a lower
frequency than they would prefer. For a specific player $i$, we call the sum of this ``leftover'' frequency from the interactions
with all $j \in \CN_i$ the \textit{slack} of player $i$, which is equal to

\[
Sl_i(\ff) = \beta_i - \sum_{\mathclap{\substack{j \in \CN_i \\ j \neq i}}} {f^*_{ij}}
\]

and in the case where $Sl_i$ is non-zero, we assign the weight of edge $\{i, i\} \in E$ to be $f^*_{ii} = Sl_i$. Furthermore, we
call the sum of all the players' frequency slack the \textit{total slack} of the game and we denote it by

\[
Sl(\ff) = \sum_{i = 1}^n {Sl_i(\ff)}
\]

If player $i$ {\em proposes} lower interaction frequency to $j$ than what $j$ proposes to $i$, we say that $i$ ``wins'' over $j$, or $j$ is in her ``win'' set.

\[
W_i(\ff) = \big\{ j \in \CN_i \: | \: j \neq i \: and \: f_{ij} < f_{ji} \big\}
\]

Similarly, if $i$'s proposal to $j$ is at least what $j$ proposes to $i$ then $i$ ``loses'' to $j$, or $j$ is in her ``lose'' set.

\[
L_i(\ff) = \big\{ j \in \CN_i \: | \: j \neq i \: and \: f_{ij} \geq f_{ji} \big\}
\]

Note that it is possible for a player $i$ to have either $W_i = \emptyset$ or $L_i = \emptyset$.

The {\em best-response} formulation of \eqref{eq:BR} raises an interesting question: What should a player do with her frequency slack, if she has any left even after she matches the frequency proposals of all her neighbors?
Specifically, what should player $i$ do when $f_{ij} \ge f_{ji}$ for all
$j \in \CN_i$ but $Sl_i(\ff) > 0$? Since no strategy in this case strictly increases $i$'s utility, all possible ways of distributing $Sl_i$ will give a best-response. To answer this question, we consider two types of players, keeping in mind the dynamical nature of the system. Consider a player $i$ who has positive slack at a given profile $\ff$.

$$
|W_i(\ff)| = 0 \: \: \text{and} \: \: Sl_i(\ff) > 0.
$$

We call $i$ \textit{pessimistic} if she decides not to spend any slack frequency on her neighbors, therefore
setting $f_{ij}=f_{ji}$ and $f_{ii} = Sl_i$. Similarly, we call $i$ \textit{optimistic}, if she decides to spend any portion
of her slack frequency on $L_i$, even though this strategy does not increase her utility right now, with the hope that
maybe, at some future time,  some player $j \in L_i(t)$ will have slack frequency and will be willing to agree to interacting
with $i$ at a higher frequency than before.

\par This distinction is important, as different strategy profiles for the players may lead to different results in the network
in terms of dynamics as well as fixed-points. More specifically, we can define two different types of equilibria for our game.

\begin{definition}[Pessimistic/Optimistic Equilibrium]\label{def:pess-opt}
A \NE $\ff$ is called a {\em pessimistic equilibrium} if for all $i, j \in \players$ such that $(i,j) \in E$, we have
$f_{ij} = f_{ji}$. It is called an {\em optimistic equilibrium} if $\exists i \in \players$ such that
$\exists j \in \CN_i$ and $f_{ij} > f_{ji}$.
\end{definition}

In words, any profile where proposals are matched on every link is a {\em pessimistic equilibrium}. On the other hand under
{\em optimistic equilibria} even though players may be proposing higher frequency to a neighbor, the neighbor does not want to
respond by increasing the frequency on the link to her. This is clearly a stricter condition to achieve compared to
{\em pessimistic equilibria}. However, as we will see in the following sections, pessimistic equilibria also turn out to be of
interest, since they possess nice convexity properties.

\subsection{Best-Response Dynamics: Sequential or Simultaneous}
\label{sec:prelim:dynam}
We analyze dynamics of the interaction in our social network for its convergence properties. Whenever a player is given an opportunity to update her strategy, it is natural for her to play a {\em best-response} against the current strategy profile of the other players (solution of \eqref{eq:BR}). Therefore, we consider the {\em best-response} (BR) dynamics under it's two natural variants, {\em simultaneous move}, and {\em sequential move}. 
Rounds are indexed by $t$ and the frequency proposal profile in round $t$ is denoted by $\ff(t)=(\ff_1(t),\dots,\ff_n(t))$. In both cases players start with certain initial proposal at $t=0$ which may be arbitrary or random. 

In the \textit{sequential} BR dynamics, in every round exactly one player updates: In round $t$, if there exists a player not playing best-response against $\ff(t)$ (in other words $\ff(t)$ is not a Nash equilibrium), then an arbitrary such player plays a best-response. That is, a player $i \in \players$ such that $\ff_i(t)$ is not a best-response against $\ff(t)$ is chosen, and then $\ff_i(t+1)$ is a BR of player $i$ against $\ff(t)$, while for all $j \neq i$, $\ff_j(t+1)=\ff_j(t)$. 

In the \textit{simultaneous move} setting, all players simultaneously update their proposal and play best-response to the strategy profile of the previous round, and inform their neighbors. Note that, there is a unique best-response for a pessimistic player, but an optimistic player $i$ may have multiple best-responses due to many possible ways of distributing her slack on the neighbors in $L_i$. In the latter case, we let the best-response be arbitrary.
By the definition of Nash equilibrium (Definition \ref{def:NE}), it follows that under both sequential and simultaneous move, the convergence points of BR dynamics are Nash equilibria, {\em i.e.,} where every player is playing a best-response to other player's strategies. 


%% file: convergence.tex

\section{Convergence Analysis and Results}
\label{sec:conv}

This section presents our main convergence results. We look at how the best-response (BR) dynamics behave, both in simultaneous
and sequential play. The convergence points of the best-response dynamics are states where no player
wants to unilaterally deviate from her strategy profile, {\em i.e.,} they are \NEa.

For the sequential move case, we show convergence of BR dynamics to a \NE in Sections \ref{sec:conv:seq:gr} and \ref{sec:conv:seq:gen}.
In Section \ref{sec:conv:seq:gr} we show that the global ranking model gives a weighted potential game, and thus
the convergence of BR dynamics follows relatively easily. To prove the convergence in general model however, we need to
analyze through a different, indirect manner. This proof is presented in Section \ref{sec:conv:seq:gen}. Finally, we show in
Section \ref{sec:conv:simul} that in the case of simultaneous play, the best-response dynamics need not converge, through
a simple counterexample.

Before we continue, we make a quick remark regarding the notation used in this section. Since frequency proposals depend
on time, i.e. the current round of our game, and all other quantities depend on the current strategy profile, we clarify
the notation used below. We use $f_{ij}(t)$ to represent the frequency proposal that player $i$ made to player $j$, at round
$t$ of the dynamics. Similarly, the frequency that $i$ and $j$ end up interacting at time $t$ is denoted by $f^*_{ij}(t)$,
i.e. $f^*_{ij}(t) = \min \{ f_{ij}(t), f_{ij}(t)\}$. For brevity, by abuse of notation, we will denote
$Sl_i(\ff(t))$, $W_i(\ff(t))$ and any other quantity that depends on the frequency profile by $Sl_i(t), W_i(t)$, etc,
respectively. Finally, due to space constraint we discuss the main ideas here, while all the missing proofs are presented
in Appendix \ref{app:conv}.
\medskip

%
\noindent{\bf Convergence in Sequential Play.}
First we consider the best-response dynamics under sequential-play and show our two main convergence results in Sections \ref{sec:conv:seq:gen} and \ref{sec:conv:seq:gen}. 
We show that the best-response dynamics converges to a \NE when players change their strategies one at a time. Our proof holds for any
general non-negative, increasing and concave utility function $u_{ij}(f^*_{ij})$, which underlines the generality and
importance of our results. Furthermore, our result is independent of the order in which the players take turns to change
strategies and relies only on the fact that each player is playing their best-response strategy that maximizes their
utility at each time step and that eventually all players get to play their turn at some point. It is also independent
of whether the players are optimistic, pessimistic or a mix of the two.

Since only one player can change their strategy at each turn, we have to clarify the time notation that will be used below.
Consider a player $i$, that makes a proposal $f_{ij}$ to $j$ at time $t_1$ and the next proposal $f'_{ij}$ of $i$ to $j$
happens at time $t_2 > t_1$. Then, we consider $f_{ij}(t) = f_{ij}$ for all times $t_1 \leq t < t_2$, and we imagine a ``jump''
in $f_{ij}(t)$ from $f_{ij}$ to $f'_{ij}$ at time $t_2$. This same logic applies not only to the players' proposals but to all
quantities defined so far.

\subsection{Convergence of the Global Ranking Model}
\label{sec:conv:seq:gr}

We start by providing strong convergence results of our global ranking model. In this section, we show that the global
ranking model is a weighted potential game, as we show the game admits a weighted potential function when for any players
$i, j$ we have $u_{ij}(f^*_{ij}) = u_{ji}(f^*_{ij})$. This condition does not imply symmetry between players'
interaction, since it may be the case that $w_{ij} \neq w_{ji}$.

\begin{theorem}[Weighted Potential Game]\label{theor:potential}
Given a game $\game$, a global ranking weighting system $R$ and a strategy profile $\ff$, if $u_{ij}(f^*_{ij}) = u_{ji}(f^*_{ij})$
for all strategy profiles $\ff$ and for all players $i, j \in \players$ where $j \in \CN_i$, then $\game$ admits a
weighted potential function

\begin{equation}\label{eq:potential-function}
\Phi(\ff) = \sum_{i \in \players} { \left( R(i) \cdot \nr(i) \cdot \sum_{j \in \CN_i} {w_{ij} u_{ij}(f^*_{ij})} \right) }
\end{equation}

and it is a weighted potential game.
\end{theorem}

The intuition behind our method is twofold. First of all, the players compute their utility based on $\ff^*$, instead
of $\ff$, meaning that their utility only depends on the frequency they end up communicating at instead of their frequency proposals.
Thus, if player $i$ changes their proposal to $j \in \CN_i$, $f^*_{ij}$ changes accordingly and the difference is the same
for both $i$ and $j$. Furthermore, recall that $w_{ij} = \frac{R(i)}{\nr(i)}$. By scaling player's $i$ utility in $\Phi$ by $\nr(i)$,
we obtain a symmetric expression for both players, which allows us to connect $i$'s effect on $j$ with $j$'s effect on $i$.

Since the global ranking model is a weighted potential game, it is well-known that the best-response dynamics converge
to a \NE \cite{monderer:96}.

\begin{corollary}\label{cor:gr-conv}
Consider a game $\game$ that admits a global ranking weighting system $R$ and $u_{ij}(f^*_{ij}) = u_{ji}(f^*_{ij})$
for all strategy profiles $\ff$ and for all players $i, j \in \players$ where $j \in \CN_i$. Then, the best-response
dynamics of $\game$ converge to a \NE.
\end{corollary}

\subsection{Convergence of the General Model}\label{sec:conv:seq:gen}

In the previous section, we saw that the best-response dynamics in the special case of the global ranking model
converges to a \NE, by constructing a potential function. In contrast, for the general model, with complete asymmetry in
weights, utility functions and behavior of agents, existence of any such potential function seems unlikely. We obtain
the convergence result for the general model in this section by an in depth analysis of the best-response dynamics.
We show that the best-response dynamics converge to a \NE for the general model as well, for any non-negative, increasing
and concave utility function. As such, our result fully characterizes the best-response dynamics for this game.

Our argument for the model's convergence is akin to using a two-stage potential-like function which decreases at each time
step and reaches a minimum which is equivalent to a \NE in our game. In order to avoid infinitesimal changes, first, we impose a reasonable
constraints on proposals of the players. Let $\eta>0$ be the minimum denominator of the resource under consideration, {\em i.e.,} $\eta=$ one second, if the resource is time. Now on we assume that any pair of players interact only at multiples of a
fixed constant $\eta > 0$. 
In other words, for all pairs of players $i, j$ and all times $t$, $f_{ij}(t) = \lambda_{ij}(t) \cdot \eta$, where $\lambda_{ij}(t) : \mathbb{N} \to \mathbb{N}$ is a pair-specific function and $\eta > 0$ is a fixed constant. In this case, it is without loss of generality to consider $\beta_i$s' as well, as multiples of $\eta$. Before proceeding with our main result, we will describe a sufficient condition for a frequency profile to be a \NE. 

\begin{lemma}\label{lem:wineq}
Consider a frequency profile $\ff$ such that, for every player $i \in \players$, $|W_i(\ff)| = 0$. Then, no player $i$
can strictly increase her utility, and $\ff$ is a \NE.
\end{lemma}

The first stage of our proof is to show that the total slack $Sl(t)$ of the game is monotonically decreasing.

\begin{lemma}\label{lem:slack-decr}
Under best-response dynamics, the total slack $Sl(t)$ is monotonically decreasing.
\end{lemma}

The basic intuition behind the proof of the above lemma is that whenever a player plays her best-response at each turn,
the total slack of the game decreases if she decreases her slack, or stays the same if she increases
her utility without decreasing her slack.

Next, we show that, as the best-response dynamics progress, $Sl(t)$ will decrease and, after some time $t$, it will
remain constant.

\begin{lemma}\label{lem:slack-stabilize}
Under best-response dynamics, if there exist a fixed constant $\eta > 0$ such that for all players $i \in \players$
and all times $t \geq 0$, there exists a $\lambda_{ij} : \mathbb{N} \to \mathbb{N}$ for all neighbors $j \in \CN_i$
such that $f_{ij}(t) = \lambda_{ij}(t) \cdot \eta$, then there exists a time $t_0$ such that $Sl(t) = Sl(t_0)$ for
all times $t \geq t_0$.
\end{lemma}

The basic idea behind the proof of the above lemma is that whenever the total slack of the game decreases, it decreases by a
constant amount, and it is also lower bounded by zero by definition. Thus, it can only decrease a finite
number of times.

Finally, in the following lemma, we argue that the stabilization of the total slack of the game is sufficient to prove
that the best-response dynamics converges to a \NE within a finite number of rounds.

\begin{lemma}\label{lem:winslack-decr}
If there exists some time $t_0$ such that $Sl(t) = Sl(t_0)$ for all times $t \geq t_0$, then the best-response dynamics
converge to a \NE within a finite number of rounds after $t_0$.
\end{lemma}

The basic intuition behind the proof of the above lemma is that the total slack does not decrease in a round only when
the chosen player redistributes her frequency, and this redistribution cannot cycle forever. Theorem \ref{theor:gen-conv}
now follows from Lemmas \ref{lem:slack-decr}, \ref{lem:slack-stabilize} and \ref{lem:winslack-decr}.

\begin{theorem}\label{theor:gen-conv}
Consider a game $\game$ where utility function $u_{ij}$'s are arbitrary non-negative, increasing, and concave. The sequential best-response dynamics of $\game$ converge to a \NE. 
\end{theorem}

\subsection{No Convergence in Simultaneous Play}
\label{sec:conv:simul}

In the previous sections, we showed that when players change their strategies one at a time, their frequency proposals
always converge under the best-response dynamics. Interestingly, the simultaneous setting is inherently different. In
other words, if at each round all the players play their best-response strategy to the strategies of their neighbors
at the previous round simultaneously, the best-response dynamics need not converge to an equilibrium. The initial
starting point of the dynamics, i.e. the starting strategy is arbitrary. To illustrate this point, we present a simple
counterexample that exhibits cyclic behavior of the dynamics. This cyclic behavior is well known, even in potential
games. One such famous example is the simple game {\em battle of the sexes} \cite{tim-book}.
\begin{example}\label{exam:nonconv}
Consider the following game $\game$ with $5$ optimistic players, $\players = \{ 1, 2, 3, 4, 5\}$, where the resource
constraint $\beta_i = 1$ is uniform for all players $i \in \players$, the utility function of each player $i$ for a
neighbor $j$ is $u_{ij}(f^*_{ij}) = f^*_{ij} (1 - f^*_{ij})$, for $0 \leq f^*_{ij} \leq \frac{1}{2}$, and
$\mathcal{G}$ is $K_5$, i.e. the complete graph with $5$ nodes. The weights between the players are represented in the
following matrix, where element $(i, j)$ is equal to $w_{ij}$

\[
\begin{bmatrix}
0 & \frac{1}{4}+\eps  & \frac{1}{4}+\eps  & \frac{1}{4}-\eps  & \frac{1}{4}-\eps  \\
\frac{1}{4}-\eps  & 0 & \frac{1}{4}+\eps  & \frac{1}{4}+\eps  & \frac{1}{4}-\eps  \\
\frac{1}{4}-\eps  & \frac{1}{4}-\eps  & 0 & \frac{1}{4}+\eps  & \frac{1}{4}+\eps  \\
\frac{1}{4}+\eps  & \frac{1}{4}-\eps  & \frac{1}{4}-\eps  & 0 & \frac{1}{4}+\eps  \\
\frac{1}{4}+\eps  & \frac{1}{4}+\eps  & \frac{1}{4}-\eps  & \frac{1}{4}-\eps  & 0
\end{bmatrix}
\]

for some $\eps > 0$. We want to force a cyclic behavior of the best-response dynamics for all players, where the
players mismatch the frequency proposals they make to each other. Specifically, we make each player propose a
slightly higher frequency to two players and a slightly lower frequency to the other two, and the weights guarantee
that, at each round, each player proposes a higher frequency to the players that propose a lower frequency to her.
Therefore, the players never match the proposals they make to each other.

Suppose that at time $t = 0$ the players calculate their best-response. We can easily see that the solution to
\eqref{eq:BR} for this game is $f_{ij} = w_{ij}$. Thus, at time $t = 0$, each player makes the following proposals to
each other, represented in the following matrix, where element $(i, j)$ is equal to $f_{ij}(0)$

\[
F(0) = \begin{bmatrix}
0 & \frac{1}{4}+\eps  & \frac{1}{4}+\eps  & \frac{1}{4}-\eps  & \frac{1}{4}-\eps  \\
\frac{1}{4}-\eps  & 0 & \frac{1}{4}+\eps  & \frac{1}{4}+\eps  & \frac{1}{4}-\eps  \\
\frac{1}{4}-\eps  & \frac{1}{4}-\eps  & 0 & \frac{1}{4}+\eps  & \frac{1}{4}+\eps  \\
\frac{1}{4}+\eps  & \frac{1}{4}-\eps  & \frac{1}{4}-\eps  & 0 & \frac{1}{4}+\eps  \\
\frac{1}{4}+\eps  & \frac{1}{4}+\eps  & \frac{1}{4}-\eps  & \frac{1}{4}-\eps  & 0
\end{bmatrix}
\]

Notice that the agreement frequencies at $t = 0$ are going to be

\[
F^*(0) = \begin{bmatrix}
0 & \frac{1}{4}-\eps  & \frac{1}{4}-\eps  & \frac{1}{4}-\eps  & \frac{1}{4}-\eps  \\
\frac{1}{4}-\eps  & 0 & \frac{1}{4}-\eps  & \frac{1}{4}-\eps  & \frac{1}{4}-\eps  \\
\frac{1}{4}-\eps  & \frac{1}{4}-\eps  & 0 & \frac{1}{4}-\eps  & \frac{1}{4}-\eps  \\
\frac{1}{4}-\eps  & \frac{1}{4}-\eps  & \frac{1}{4}-\eps  & 0 & \frac{1}{4}-\eps  \\
\frac{1}{4}-\eps  & \frac{1}{4}-\eps  & \frac{1}{4}-\eps  & \frac{1}{4}-\eps  & 0
\end{bmatrix}
\]

Now, all players compute their best-response strategies simultaneously. We observe that for player
$i \in \{1, 2, 3, 4, 5 \}$, their outcome sets are $L_i(0) = \{ i \: mod \: 5 + 1, (i+1) \: mod \: 5 + 1 \}$
and $W_i(0) = \{ (i+2) \: mod \: 5 + 1, (i+3) \: mod \: 5 + 1 \}$. Furthermore, each player has exactly
$4 \eps $ slack frequency. Recall that each player $i$ is optimistic and the weights for both players in
$W_i(0)$ are equal. Since $i$ has slack, their best-response, as calculated by \eqref{eq:BR}, is to allocate
$2 \eps$ of their slack to each player in $W_i(0)$.

Therefore, at time $t = 1$, the frequency proposals of every player are

\[
F(1) = \begin{bmatrix}
0 & \frac{1}{4}-\eps  & \frac{1}{4}-\eps  & \frac{1}{4}+\eps  & \frac{1}{4}+\eps  \\
\frac{1}{4}+\eps  & 0 & \frac{1}{4}-\eps  & \frac{1}{4}-\eps  & \frac{1}{4}+\eps  \\
\frac{1}{4}+\eps  & \frac{1}{4}+\eps  & 0 & \frac{1}{4}-\eps  & \frac{1}{4}-\eps  \\
\frac{1}{4}-\eps  & \frac{1}{4}+\eps  & \frac{1}{4}+\eps  & 0 & \frac{1}{4}-\eps  \\
\frac{1}{4}-\eps  & \frac{1}{4}-\eps  & \frac{1}{4}+\eps  & \frac{1}{4}+\eps  & 0
\end{bmatrix}
\]

which again yields $F^*(1) = F^*(0)$. However, note that for all players, $L_i(1) = W_i(0)$ and $W_i(1) = L_i(0)$. We apply
the previous argument now for $t = 2$ and get $F(2) = F(0)$, which shows the existence of a cycle of proposals between all
players, implying our model need not converge in the simultaneous setting.
\end{example}

%% file: game-properties.tex

\section{Properties of the Game}
\label{sec:proper}

In this section we provide several properties of our game's best-response dynamics and equilibria. First,
we show that the global optimum of our game is also a \NE, which implies that the Price of Stability (PoS) is $1$.
We also provide an example with unbounded Price of Anarchy (PoA), which demonstrates that the social welfare of the
\NEa can vary significantly. Next, we show that the set of \NEa for our game is connected and, more importantly,
the set of pessimistic \NEa is convex. Finally, we fully characterize each player's best-response through the
well-known Karush-Kuhn-Tucker (KKT) conditions for local optimality \cite{boyd2004convex}, which provides better
intuition as to how each player calculates her best-response. All the missing proofs of this section are presented
in Appendix \ref{app:proper}.

\subsection{Price of Anarchy and Stability}
\label{sec:proper:pos-poa}

\par In this section, we focus on the quality of \NEa. We show that the optimal frequency distribution for all players is
also a \NE. On the other hand, we also show that there exist arbitrarily low quality \NEa, even for the simpler case of
uniform resource constraints among all players. We quantify these observations through the well-known concepts of the
{\em Price of Anarchy (PoA)} \cite{poa} and {\em Price of Stability (PoS)} \cite{pos} respectively. These results show
that there is a significant difference between the Price of Anarchy and the Price of Stability in our game.

We first provide a definition of the Price of Anarchy and the Price of Stability for our game.

\begin{definition}[Price of Anarchy and Price of Stability]\label{def:poa-pos}
Consider a game $\game$, where $OPT$ is the frequency profile that maximizes the social welfare, and
$\mathcal{C^{\text{\textit{eq}}}}$ is the set of all frequency profiles that are \NEa. Then, the
{\em Price of Anarchy (PoA)} is defined as

\[
PoA = \sup_{\ff \in \mathcal{C^{\text{\textit{eq}}}}} {\frac{SW(OPT)}{SW(\ff)}}
\]

while the {\em Price of Stability (PoA)} is defined as

\[
PoS = \inf_{\ff \in \mathcal{C^{\text{\textit{eq}}}}} {\frac{SW(OPT)}{SW(\ff)}}
\]
\end{definition}

We first want to analyze how good a \NE can be. Theorem \ref{theor:pos} shows that
the optimal solution profile of our game is also a \NE. The optimal solution profile of our game
can be seen as the solution to the following (global) convex program

\begin{equation}\label{eq:gcp}
\begin{array}{lll}
max  & \qquad \displaystyle{\sum_{i = 1}^n {\sum_{\mathclap{\substack{j \in \CN_i \\ j \neq i}}} {w_{ij} u_{ij}(f^*_{ij}) } } } & \\
s.t.    & \qquad f^*_{ij} \leq f_{ij} \quad \& \quad f^*_{ij} \le f_{ji} \qquad & \forall i, \forall j \in \CN_i \\
      	& \qquad \sum_{j \in \CN_i} {f_{ij}} \leq \beta_i & \forall i \\
      	& \qquad f^*_{ij} \geq 0 & \forall i, \forall j \in \CN_i \\
      	& \qquad f_{ij} \geq 0 & \forall i, \forall j \in \CN_i \\
\end{array}
\end{equation}

Note that, unlike \eqref{eq:BR} where $f_{ji}$ was a constant for $i$'s best-response, it is now a variable in this
program. It is clear that the solution to the above program is the optimal solution profile of our game, i.e. the
frequency profile $\tff$ that maximizes the social welfare. Recall that Lemma \ref{lem:wineq} describes a sufficient
condition for a frequency profile to be a \NE. We now look into how we can transform any frequency profile into a
pessimistic \NE, with equal social welfare, through a simple process of making each player match the proposals of
her neighbors.

\begin{lemma}\label{lem:match}
Consider a frequency profile $\ff$. Then, we can construct a frequency profile $\ff'$ such that $SW(\ff') = SW(\ff)$ and
$\ff'$ is a pessimistic \NE.
\end{lemma}

Lemma \ref{lem:match} is enough to guarantee that the optimal solution of our game is also a \NE.

\begin{theorem}\label{theor:pos}
Let $\game$ be a game. Then, there exists an optimal strategy profile that maximizes the social welfare of
$\game$ and is also a \NE.
\end{theorem}
\begin{proof}
Let $\tff$ be the solution of \eqref{eq:gcp}. From Lemma \ref{lem:match}, we can construct a new frequency
profile $\tff'$ from $\tff$ such that $SW(\tff') = SW(\tff)$ and $\tff'$ is a pessimistic \NE. Since $\tff'$
has the same social welfare as $\tff$, we understand that $\tff'$ is also a solution to the global convex program
\eqref{eq:gcp} that maximizes the social welfare. Thus, $\tff'$ both maximizes the social welfare of $\game$ and
is also a \NE.
\end{proof}

Theorem \ref{theor:pos} leads us to the following corollary.

\begin{corollary}\label{cor:pos}
The Price of Stability of a game $\game$ is $1$.
\end{corollary}

Next, we look at how bad can equilibria be for our game. Unfortunately, the next theorem demonstrates that there
exist \NEa with arbitrarily bad social welfare, even for the simple case where all players have the same resource
constraint. 

\begin{theorem}\label{theor:poa}
The Price of Anarchy of a game $\game$ is unbounded.
\end{theorem}

\subsection{Properties of the set of Nash Equilibria}
\label{sec:proper:setNE}

In this section we look at the convexity and connectedness of the set of \NEa for our game. Specifically, we
show that the set of pessimistic \NEa is convex, which implies that the set of all \NEa of our game is connected.
We start with the following lemma.

\begin{lemma}\label{lem:opt-pes}
Every optimistic \NE $\ff$ can be transformed into a pessimistic \NE $\ff'$ with the same social welfare. Furthermore,
every frequency profile that is a convex combination of $\ff$ and $\ff'$ is also an optimistic \NE.
\end{lemma}

We now show that the set of pessimistic \NEa is convex.

\begin{theorem}\label{theor:pes-cvx}
Consider a game $\game$ with two pessimistic \NEa $\ff$ and $\ff'$. Any convex combination of $\ff$ and $\ff'$ is
also a pessimistic \NE.
\end{theorem}
\begin{proof}
The proof is almost identical to the proof of Lemma \ref{lem:opt-pes}. We want to show that every convex combination of
$\ff$ and $\ff'$ is also a pessimistic \NE. Let $\alpha \in [0, 1]$, and $\ff'' = \alpha \ff + (1 - \alpha) \ff'$. Consider
a pair of players $i, j \in \players$. Since $\ff$ and $\ff'$ are pessimistic \NEa, we have $f_{ij} = f_{ji}$ and
$f'_{ij} = f'_{ji}$. Thus

\[
f''_{ij} = \alpha f_{ij} + (1 - \alpha) f'_{ij} = \alpha f_{ji} + (1 - \alpha) f'_{ji} = f''_{ji}
\]

and $i,j$ make matching frequency proposals to each other at $\ff''$. Since this holds for every such pair
$i, j \in \players$, we understand that $\ff''$ is a pessimistic \NE.
\end{proof}

\begin{corollary}\label{cor:cvx-set}
The set of pessimistic \NEa of a game $\game$ is convex.
\end{corollary}

The following corollary now follows from Lemma \ref{lem:opt-pes} and Theorem \ref{theor:pes-cvx}.

\begin{corollary}\label{cor:conn-set}
The set of \NEa of a game $\game$ is connected.
\end{corollary}

Finally, we provide a complete characterization of a player's best-response in Appendix \ref{app:br}.

%% file: experiments.tex

\section{Experimental Results}
\label{sec:exper}
\begin{figure}[!htbp]
 \includegraphics[width=0.75\textwidth]{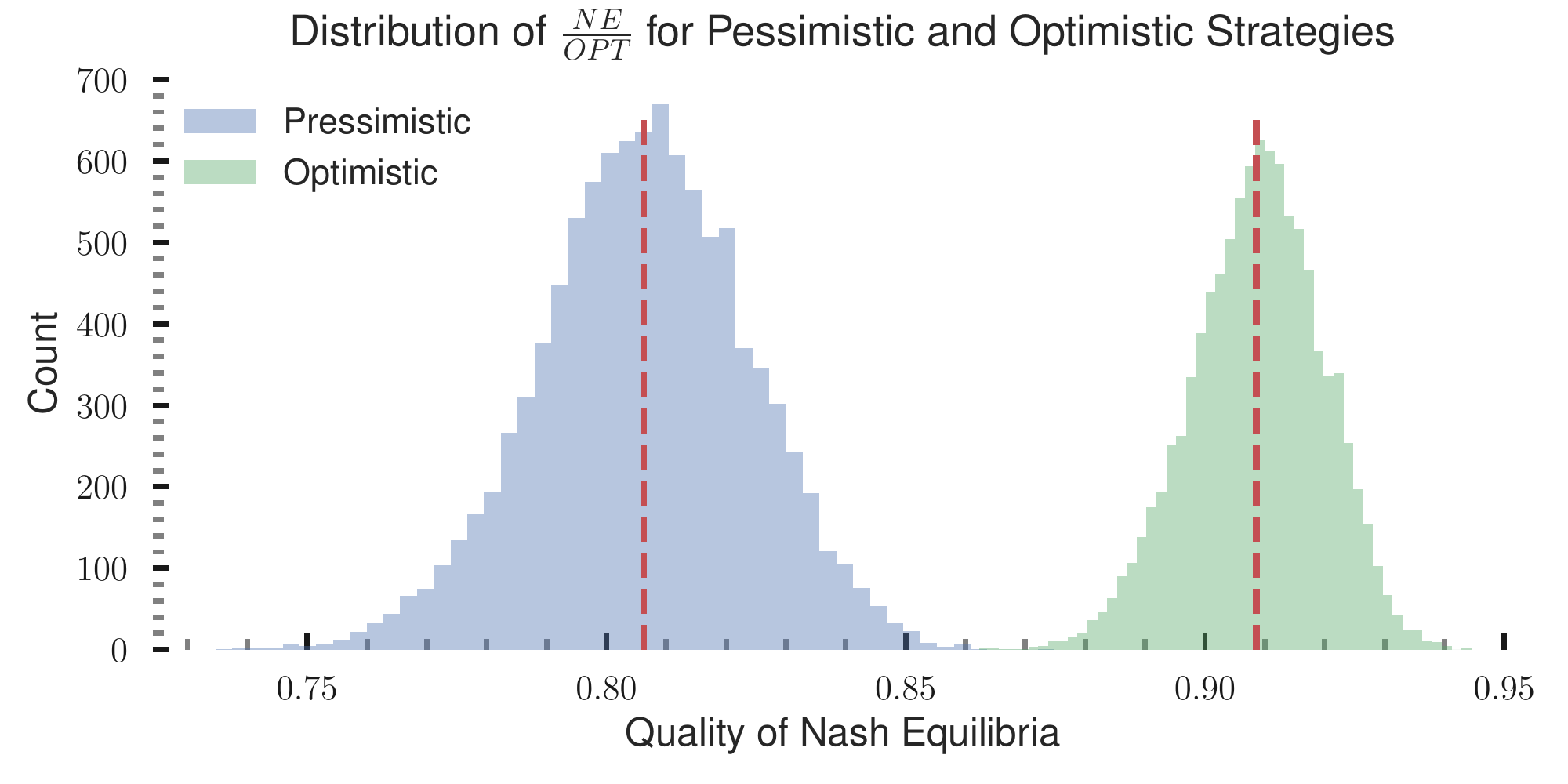}
 \caption{We show the difference between the quality of Nash Equilibria, when players adopt either pessimistic (i.e. do not redistribute slack), or optimistic (distribute slack over winning set) strategies (see ~\Cref{sec:prelim:opt-pes}). We conducted 10,000 simulations on a stylized 4-regular graph with random weights, uniform endowment ($\beta=1000$) and random initial proposal frequencies. The histogram ($y$-axis shows the number of simulations that converge to a particular bin; $x$-axis is ratio of NE quality to OPT, appropriately binned) shows that the optimistic strategies outperform pessimistic strategies, with a higher mean quality (dashed lines) and slightly lower variance.}
 \label{fig:comparision}
\end{figure}

We present results of simulations of best-response dynamics on a stylized graph. We would like to compute the distribution of Nash Equilibria when players adopt either the pessimistic or the optimistic strategies to distribute their slack (see ~\Cref{sec:prelim:opt-pes}). The pessimistic strategy for a player $i$ involves not distributing slack, while the optimistic strategy involves re-distributing slack over the winning set (i.e. set of neighbors $j \in \mathcal{N}_i$ whose proposals are higher than $i$'s proposals to her neighbors).

For the simulation, we use a $10 \times 10$ grid (i.e. a 4-regular graph), with uniform resource constraint of $\beta=1000$, random edge weights. The simulation involves using random frequency proposals to start and then compute best-response dynamics using a random sequence of player updates. We ran 10,000 simulations and noted the quality of Nash Equilibrium for each run for each strategy. Note that, in our experiments, either all players use the pessimistic strategy or all players use the optimistic strategy. We show the results in~\Cref{fig:comparision}, by creating a histogram of the ratio of the Nash Equilibrium quality to the global optimum, for each strategy. The results show that on average the quality of Equilibria for the optimistic strategy ($\mu_o=0.908$) is significantly better than the pessimistic strategy ($\mu_p=0.806$), and the quality of equilibria for the optimistic strategy has slightly lower variance ($\sigma_o=0.011, \sigma_p=0.017$). It is gratifying that the optimistic strategy does so well. Notice that the distribution is unimodal in for each strategy---we conjecture that this may be an outcome of the connectedness of the Nash Equilibria (see Corollary~\ref{cor:conn-set}).

%% file: conclusion.tex

\section{Conclusion}
\label{sec:Conclusion}
The problem of time allocation is one of longstanding interest to Economics and to Sociology given its importance to economic output and to the sustenance of social networks. In this paper, we formally studied the resource allocation game, where agents have private, asymmetric interaction preferences and make decisions on time allocation over their social network. The game was challenging to analyze since it is not in general, a weighted potential game, and its best response dynamics are not differentiable. First we showed that a restricted subclass of games where the interaction preferences are  related to the social rank is a weighted potential game. Then, for the general case, we used a novel two-level potential function approach to show that the best response dynamics converge to Nash Equilibrium. Our proof is general, and makes no assumptions on the form of the utility function beyond that it is concave, increasing and non-negative, which are reasonable and standard assumptions. We showed that the Price of Anarchy is unbounded, and that the Price of Stability is unity. Furthermore, we showed that the Nash Equilibria form a connected set. Towards understanding the quality \NE where best response converges, 
through extensive simulation of a stylized graph, we showed that the distribution of quality of Nash Equilibria are unimodal, which we conjecture is related to the connectedness of Nash Equilibria.

We identify two assumptions that limit the generalizability of our results. Our analysis of time focuses on costly communication (e.g. a conversation over a phone, or meeting in person). In online social networks, communication may be asymmetric---agent $i$ may send more messages to agent $j$ than does agent $j$ send to $i$. Second we assume that strength of the tie doesn't change over time---$w_{ij}$ remains the same. In real-world networks, tie strengths improve and degrade over time~\cite{Miritello2013}. We can incorporate weight changes by allowing weight update of $w_{ij}$ depending on how the neighbor $j$ of agent $i$ reciprocates to her proposals. 

Finally, it would be interesting to further understand the quality of \NE to which the resource allocation game converges, through the lens of average price of anarchy \cite{PP16}.


%% file: appendix.tex

\appendix

{\Huge \textbf{Appendix}}

\section{Missing Proofs of Section \ref{sec:conv}}\label{app:conv}

\subsection{Proof of Theorem \ref{theor:potential}}\label{app:potential}

Because the game adopts global-ranking weighting scheme, every weight $w_{ij}$ can be written as

\[
w_{ij} = \frac{R(j)}{\sum_{k \in \CN_i} R(k)} = \frac{R(j)}{\nr(i)}
\]

Recall the potential function we proposed in \eqref{eq:potential-function}. If player $i$ deviates from $\ff_i$, which
is the vector of frequency proposals to all $j \in \CN_i$, to $\ff'_i$, then

\begin{align*}
\Phi(\ff_{-i}, \ff'_i) - \Phi(\ff) & =
\sum_{i \in \players} { \left( R(i) \nr(i) \sum_{j \in \CN_i} {w_{ij} (u_{ij}(f'^*_{ij}) - u_{ij}(f^*_{ij}))} \right) } \\
& = \sum_{i \in \players} { \left( R(i)  \sum_{j \in \CN_i} { R(j) (u_{ij}(f'^*_{ij}) - u_{ij}(f^*_{ij}))} \right) }
\end{align*}

Notice that only $i$'s neighbors are affected by this deviation, therefore

\begin{align*}
\Phi(\ff_{-i}, \ff'_i) - \Phi(\ff) & = \sum_{j \in \CN_i} {R(i) R(j) (u_{ij}(f'^*_{ij}) - u_{ij}(f^*_{ij}))} \\
& + \sum_{j \in \CN_i} {R(j) R(i) (u_{ji}(f'^*_{ij}) - u_{ji}(f^*_{ij}))}
\end{align*}

Since $u_{ij}(f^*_{ij}) = u_{ji}(f^*_{ij})$, we have

\begin{align*}
\Phi(\ff_{-i}, \ff'_i) - \Phi(\ff) & = \sum_{j \in \CN_i} {R(i) R(j) (u_{ij}(f'^*_{ij}) - u_{ij}(f^*_{ij}))} \\
                                   & + \sum_{j \in \CN_i} {R(j) R(i) (u_{ij}(f'^*_{ij}) - u_{ij}(f^*_{ij}))} \\
                                   & = 2 R(i) \sum_{j \in \CN_i} {R(j) (u_{ij}(f'^*_{ij}) - u_{ij}(f^*_{ij}))} \\
                                   & = 2 R(i) \nr(i) \sum_{j \in \CN_i} {w_{ij} (u_{ij}(f'^*_{ij}) - u_{ij}(f^*_{ij}))} \\
                                   & = 2 R(i) \nr(i) u_i(\ff'_i) - u_i(\ff_i)
\end{align*}

where $u_i(\ff)$ denotes the total utility of player $i$ at strategy profile $\ff$.

The difference of in $\Phi$ after the move of player $i$ is the difference of $i$'s utility, scaled by a constant.
Therefore $\game$ is a weighted potential game and always admits a \NE \cite{monderer:96}.

\subsection{Proof of Lemma \ref{lem:wineq}}\label{app:wineq}

We know that $\forall i \in \players$, $|W_i(\ff)| = 0$. This implies that $\forall j \in \CN_i$,
$f_{ij} \geq f_{ji}$. Therefore, $i$ is matching the frequency proposals of all her neighbors, meaning that when $i$
calculates her best-response, she satisfies all constraints of the form $f^*_{ij} \leq f_{ji}$ in \eqref{eq:BR} for
some $j \in \CN_i$ with equality. It follows that $i$ cannot increase her utility by unilaterally deviating from $\ff$.
Since this condition holds for every $i \in \players$, we understand that $\ff$ is a \NE.

\subsection{Proof of Lemma \ref{lem:slack-decr}}\label{app:slack}

Consider that at time $t$ it is $i$'s turn in our model. Obviously, we have $Sl_j(t+1) = Sl_j(t)$ for all
$j \notin \CN_i$, since $i$'s turn did not affect them at all. We first prove the following claims which relate the
slack of $i$ at time $t$, before she plays her best-response, and at time $t+1$, after $i$ has made new frequency
proposals according to her best-response.

\begin{claim}\label{clm:empty-win}
If at time $t$ it is $i$'s turn to play her best-response strategy and $|W_i(t)| = 0$, then $Sl_i(t+1) = Sl_i(t)$.
Moreover, $Sl(t+1) = Sl(t)$.
\end{claim}
\begin{proof}
Since $|W_i(t)| = 0$, $i$ is matching the proposal of all of her neighbors. Thus, every strategy $\bm{f}_i(t+1)$
where $f_{ij}(t+1) \geq f_{ji}(t)$ for all $j \in \CN_i$ is a best-response strategy. This implies that in any
best-response strategy, $f^*_{ij}(t+1) = f^*_{ij}(t)$, for all $j \in \CN_i$. We know that
$Sl_i(t) = \beta_i - \sum_{j \in \CN_i} {f^*_{ij}(t)}$, and also
\[
Sl_i(t+1) = \beta_i - \sum_{j \in \CN_i} {f^*_{ij}(t+1)} = \beta_i - \sum_{j \in \CN_i} {f^*_{ij}(t)} = Sl_i(t)
\]
where the second equality follows from the fact that the utility function of $i$ is increasing.

Finally, since $i$'s proposals did not force any change in the interaction frequency with any player $j \in \CN_i$,
we have
\begin{align*}
Sl_j(t+1) & = \beta_j - \sum_{k \in \CN_j} {f^*_{jk}(t+1)} = \beta_j - \sum_{k \in \CN_j \setminus \{i\}} {f^*_{jk}(t+1)} - f^*_{ji}(t+1) \\
& = \beta_j - \sum_{k \in \CN_j \setminus \{i\}} {f^*_{jk}(t)} - f^*_{ji}(t) = \beta_j - \sum_{k \in \CN_j} {f^*_{jk}(t)} = Sl_j(t)
\end{align*}

Combining the two observations above, we understand that $Sl(t+1) = Sl(t)$.
\end{proof}

\begin{claim}\label{clm:personal}
If at time $t$ it is $i$'s turn to play her best-response strategy, $|W_i(t)| > 0$ and $Sl_i(t) \geq \eta$,
then, $Sl(t+1) < Sl(t)$.
\end{claim}
\begin{proof}
Notice that the proposals $\bm{f}_i(t)$ of $i$ at time $t$, before $i$'s turn, is a feasible strategy. It follows easily
that the utility of the best-response strategy of $i$ is at least as good as the utility of $\bm{f}_i(t)$. Suppose
that $Sl_i(t+1) \geq Sl_i(t)$. We have
\begin{align}\label{eq:app:1}
& Sl_i(t+1) \geq Sl_i(t) \Leftrightarrow \beta_i - \sum_{j \in \CN_i} {f^*_{ij}(t+1)} \geq \beta_i - \sum_{j \in \CN_i} {f^*_{ij}(t)} \Leftrightarrow \sum_{j \in \CN_i} {f^*_{ij}(t)} \geq \sum_{j \in \CN_i} {f^*_{ij}(t+1)} \nonumber \\
& \Leftrightarrow \sum_{j \in W_i(t)} {f_{ij}(t)} + \sum_{j \notin W_i(t)} {f_{ji}(t)} \geq \sum_{j \in W_i(t+1)} {f_{ij}(t+1)} + \sum_{j \notin W_i(t+1)} {f_{ji}(t+1)}
\end{align}
Now, since at time $t$ only $i$ played her best response, it must be that
\begin{align}\label{eq:app:2}
& \sum_{j \in \CN_i} {f_{ji}(t)} = \sum_{j \in \CN_i} {f_{ji}(t+1)} \nonumber \\
& \Leftrightarrow \sum_{j \in W_i(t)} {f_{ji}(t)} + \sum_{j \notin W_i(t)} {f_{ji}(t)} = \sum_{j \in W_i(t+1)} {f_{ji}(t+1)} + \sum_{j \notin W_i(t+1)} {f_{ji}(t+1)} \nonumber \\
& \Leftrightarrow \sum_{j \notin W_i(t)} {f_{ji}(t)} = \sum_{j \in W_i(t+1)} {f_{ji}(t+1)} + \sum_{j \notin W_i(t+1)} {f_{ji}(t+1)} - \sum_{j \in W_i(t)} {f_{ji}(t)}
\end{align}
Substituting \eqref{eq:app:2} into \eqref{eq:app:1} gives us
\begin{align}\label{eq:app:3}
& \sum_{j \in W_i(t)} {f_{ij}(t)} + \sum_{j \in W_i(t+1)} {f_{ji}(t+1)} + \sum_{j \notin W_i(t+1)} {f_{ji}(t+1)} - \sum_{j \in W_i(t)} {f_{ji}(t)} \nonumber \\
& \geq \sum_{j \in W_i(t+1)} {f_{ij}(t+1)} + \sum_{j \notin W_i(t+1)} {f_{ji}(t+1)} \nonumber \\
& \Leftrightarrow \sum_{j \in W_i(t+1)} {f_{ji}(t+1)} - \sum_{j \in W_i(t+1)} {f_{ij}(t+1)} \geq \sum_{j \in W_i(t)} {f_{ji}(t)} - \sum_{j \in W_i(t)} {f_{ij}(t)}
\end{align}
Since $|W_i(t)| > 0$, we have by definition that
\begin{align}\label{eq:app:4}
& \exists \: j \in \CN_i : f_{ji}(t) > f_{ij}(t) \Leftrightarrow \sum_{j \in W_i(t)} {f_{ji}(t)} > \sum_{j \in W_i(t)} {f_{ij}(t)}
\end{align}
Thus, combining \eqref{eq:app:3} and \eqref{eq:app:4}, we get
\begin{align}\label{eq:app:4}
& \sum_{j \in W_i(t+1)} {f_{ji}(t+1)} - \sum_{j \in W_i(t+1)} {f_{ij}(t+1)} > 0
\end{align}
This implies that there exists at least one player $j \in \CN_i$ such that $f_{ij}(t+1) < f_{ji}(t+1)$, 
which directly gives us that $|W_i(t+1)| > 0$. Furthermore, since $Sl_i(t) \geq \eta$ and
$Sl_i(t+1) \geq Sl_i(t)$, we have that $Sl_i(t+1) \geq \eta$, which gives us that $i$ at time $t+1$ has both
a non-empty winning set and slack at least $\eta$. Since $i$'s utility function is increasing, she can
``spend'' her slack on her winning set and increase her utility which directly contradicts that $i$ played her
best-response strategy. Thus, $Sl_i(t+1) < Sl_i(t)$.

\end{proof}

\begin{claim}\label{clm:reshuffle}
If at time $t$ it is $i$'s turn to play her best-response strategy, $|W_i(t)| > 0$ and $Sl_i(t) = 0$,
then $Sl(t+1) \leq Sl(t)$.
\end{claim}
\begin{proof}
Suppose that $i$'s best-response dictates that she has to increase her proposal to a set of agents $S \subseteq W_i(t)$.
The slack of $i$ is zero, and thus, in order for $i$ to increase her proposals to $S$, she will have to decrease her
proposal to a set $S'$ of agents in her neighborhood. Therefore, we say that $i$ has to perform some
\textit{reshuffling} of her proposals. First, note that $Sl_j(t+1) = Sl_j(t)$ for all players $j \notin S \cup S'$,
since they are not affected by $i$'s reshuffling.

If $i$ decreases her proposals to $S'$ more than how much she wants to increase her proposals to $S$, her slack will
increase, which implies that this strategy is not a best-response. Thus, we have that
\[
\sum_{j \in S'} {\left(f_{ij}(t) - f_{ij}(t+1)\right)} + Sl_i(t) = \sum_{j \in S} {\left(f_{ij}(t+1) - f_{ij}(t)\right)} + Sl_i(t+1)
\]
By Claim \ref{clm:personal}, we have that $Sl_i(t+1) \leq Sl_i(t)$, which implies
\[
\sum_{j \in S'} {\left(f_{ij}(t) - f_{ij}(t+1)\right)} \leq \sum_{j \in S} {\left(f_{ij}(t+1) - f_{ij}(t)\right)}
\]
Notice here that upon this reshuffling, the interaction frequency of $i$ with any agent in $S \cup S'$ will change
immediately. Specifically, we have
\[
\begin{cases}
f^*_{ij}(t+1) = f_{ij}(t+1) < f^*_{ij}(t) & \forall j \in S' \\
f^*_{ij}(t+1) = f_{ij}(t+1) > f^*_{ij}(t) & \forall j \in S
\end{cases}
\]
Thus, we get that
\begin{align*}
\sum_{j \in S'} {\left(f^*_{ij}(t) - f^*_{ij}(t+1)\right)} \leq \sum_{j \in S} {\left(f^*_{ij}(t+1) - f^*_{ij}(t)\right)} \Leftrightarrow \\
\sum_{j \in S \cup S'} {f^*_{ij}(t)} \leq \sum_{j \in S \cup S'} {f^*_{ij}(t+1)} \Leftrightarrow \\
\sum_{j \in S \cup S'} {\left(\beta_j  - f^*_{ij}(t)\right)} \geq \sum_{j \in S \cup S'} {\left(\beta_j - f^*_{ij}(t+1)\right)} \Leftrightarrow \\
\sum_{j \in S \cup S'} {Sl_j(t)} \geq \sum_{j \in S \cup S'} {Sl_j(t+1)}
\end{align*}
Using this fact along with Claim \ref{clm:personal} and the fact that agents not in $S \cup S'$ are unaffacted by $i$'s
reshuffling, we get
\[
\begin{array}{clll}
Sl(t) & = Sl_i(t) & + \sum_{j \notin S \cup S'} {Sl_j(t)} & + \sum_{j \in S \cup S'} {Sl_j(t)} \\
& \geq Sl_i(t+1) & + \sum_{j \notin S \cup S'} {Sl_j(t+1)} & + \sum_{j \in S \cup S'} {Sl_j(t+1)} \\
& = Sl(t+1)
\end{array}
\]
\end{proof}

We conclude that since at every time $t$, $Sl(t)$ either decreases or remains constant, $Sl(t)$ is monotonically
decreasing.

\subsection{Proof of Lemma \ref{lem:slack-stabilize}}\label{app:slack-stabilize}

We will prove the lemma by contradiction. Assume that for any time $t$, there exists a time $t' > t$ such that
$Sl(t') < Sl(t)$. Since every new proposal at $t'$ must be at least $\eta$ greater or lesser than the previous
proposal at time $t'-1$, we know that $Sl(t'-1) - Sl(t') \geq \eta$. Since the total slack is monotonically
decreasing by Lemma \ref{lem:slack-decr}, we have that $Sl(t'-1) \leq Sl(t)$, and thus $Sl(t) - Sl(t') \geq \eta$.
Now, the argument can be repeated once again. We know that after at most $\frac{\sum_{i = 1}^n {\beta_i}}{\eta} + 1$
such repeats of this argument, we will reach a time $t^*$ such that $Sl(t^*) < 0$, and we arrive at a contradiction,
based on the definition of the total slack of the game. Thus, we conclude that there exists a time $t_0$ such that
$Sl(t) = Sl(t_0)$ for all times $t \geq t_0$.

\subsection{Proof of Lemma \ref{lem:winslack-decr}}\label{app:winslack}

Suppose we are at time $t \geq t_0$, after the total slack has stabilized. Since $\eta$ is the minimum denomination of the resource, it is without loss of generality to assume that all $\beta_i$'s are multiple of $\eta$. Furthermore, players' proposals are in multiples of $\eta$, and therefore for any player $i$, $Sl_i(t)$ is either $0$ or at least $\eta$ at any time $t$.

First of all, note that at time $t$ there must exist some player $i$ with $|W_i(t)| > 0$, because if $W_i(t) = 0$
for all players $i \in \players$, then, by Lemma \ref{lem:wineq}, we have reached a \NE. Also, note that for every
player $i$ that has $|W_i(t)| > 0$, it must be that $Sl_i(t)=0$. Otherwise, by Claim \ref{clm:personal},
we get that the total slack decreases, which contradicts our hypothesis. Thus, we can
partition the players into two sets, $V_1(t)$ and $V_2(t)$, where $V_1(t)$ is the set of players $i$ that have
$|W_i(t)| = 0$ and $Sl_i(t) \ge 0$, 
while $V_2(t)$ is the set of players $j$ that have $|W_j(t)| > 0$ and $Sl_j(t)=0$. Note that, sets $V_1(t)$ and $V_2(t)$ covers all the players.

We first show that players in $V_1(t)$ are always playing best-response after time $t_0$. 

\begin{claim}\label{clm:v1stasis}
For all $t\ge t_0$, and for all $i \in V_1(t)$, $\ff_i(t)$ is a best-response against $\ff(t)$.
\end{claim}
\begin{proof}
Since $i \in V_1(t)$, we understand that $|W_i(t)| = 0$. By the definition of $W_i(t)$, we have that
$f_{ij}(t) \geq f_{ji}(t)$, for all players $j \in \CN_i$. Thus, if we try to compute best-response of player $i$ against $\ff(t)$ using the (local) convex program \eqref{eq:BR}, then all inequalities of the form $f^*_{ij} \le f_{ji}(t)$ will be tight, since $i$ has extra budget but no neighbor to spend on to and $u_{ij}$s are increasing in $f^*_{ij}$. Thus, it follows that $\ff_i(t)$ is a best-response of player $i$ against $\ff(t)$. 
%
\end{proof}

The above claim implies the strategies of the agents in $V_1(t)$ do not change at time $t$. 

\begin{claim}\label{clm:v1non-decr}
$V_1(t) \subseteq V_1(t+1)$.
\end{claim}
\begin{proof}
By Claim \ref{clm:v1stasis}, players in $V_1(t)$ will not change their strategies at time $t$. Thus, the 
player who changes strategy and plays a best-response at time $t$ has to belong to $V_2(t)$.
Let $i \in V_2(t)$ be a player that changes her strategy between times $t$ and $t+1$, {\em i.e.,} $\forall j \neq i$ $\ff_j(t+1)=\ff_j(t)$. 

Since $|W_i(t)| > 0$, $i$ does not have enough budget to match every
proposal made to her, we understand that $i$'s best-response strategy forces $f_{ij}(t) \leq f_{ji}(t)$
for all $j \in \CN_i$. Since $Sl_i(t)=0$, her best-response strategy for $t+1$ can never be to exceed a
proposal made to her, since this yields no utility, while there is utility to be gained for $i$ by trying to
match the players in $W_i(t)$. Thus, we understand that $i$'s best-response strategy forces
$f_{ij}(t+1) \leq f_{ji}(t+1)$ for all $j \in \CN_i$. Since there is no $j \in \CN_i$ such that
$f_{ij}(t+1) > f_{ji}(t+1)$, every player $k \in V_1(t)$ will still have $|W_k(t+1)| = 0$, and thus $k \in V_1(t+1)$.
\end{proof}

Next, we show that whenever a player in $V_2(t)$ increases her utility, it increases by at least some fixed constant.

\begin{claim}\label{cl:util-inc}
If at time $t$ it is $i$'s turn to play her best-response strategy and $i \in V_2(t)$, then either
$u_i(\ff(t+1)) = u_i(\ff(t))$, or
\[
u_i(\ff(t+1)) \geq u_i(\ff(t)) + \Delta_i
\]
where $\Delta_i > 0$ is a fixed constant.
\end{claim}
\begin{proof}
At time $t$, $i$ calculates her best-response strategy $\ff_i(t+1)$ by solving \eqref{eq:BR}.
Since the solution to \eqref{eq:BR} maximizes $i$'s utility at time $t+1$, we know that
\[
u_i\left(\ff(t+1)\right) - u_i\left(\ff(t)\right) \geq 0
\]
Clearly, if $u_i\left(\ff(t+1)\right) = u_i\left(\ff(t)\right)$, then $\ff_i(t)$ is also a best-response strategy
at time $t+1$, and $i$ can simply not change strategies and get the same utility. Suppose now that
$u_i\left(\ff(t+1)\right) > u_i\left(\ff(t)\right)$. We know that $f_{ij}(t) = \lambda_{ij}(t) \cdot \eta$ and
$f_{ij}(t+1) = \lambda_{ij}(t+1) \cdot \eta$. Thus, there exist only ${\left( \frac{\beta_i}{\eta} \right)}^2$ choices
for the pair $(f_{ij}(t), f_{ij}(t+1))$. By the same argument, there exist only
${\left( \frac{\beta_i}{\eta} \right)}^{2 |\CN_i|}$ choices for the pair of strategies $(\ff_i(t), \ff_i(t+1))$, which
are finite. Recall that $\mathcal{F}_i$ denotes the set of all possible strategies of $i$. Let $\mathcal{F}^{2>}_i$ be
the set of all possible pairs of strategies of $i$ such that the first has higher utility than the second. In other
words,
\[
\mathcal{F}^{2>}_i = \left\{ \left( \ff_i, \ff'_i \right) \ \bigg| \ \ff_i, \ff'_i \in \CF_i \ \ \text{and} \ \ u_i(\ff) > u_i(\ff') \right\}
\]
We define
\[
\Delta_i = \: \min_{\left(\ff_i, \ff'_i\right) \in \mathcal{F}^{2>}_i} \: {u_i(\ff) - u_i(\ff')}
\]
This is the minimum possible increase in utility between two strategy profiles of $i$, and it is a fixed positive
constant, since ${|\mathcal{F}_i|}^2$ is finite. Clearly, for any pair of strategy profiles $(\ff_i(t), \ff_i(t+1))$
of $i$, if $u_i\left(\ff(t+1)\right) > u_i\left(\ff(t)\right)$, then $u_i(\ff(t+1)) \geq u_i(\ff(t)) + \Delta_i$.
\end{proof}

Finally, we show that utility of players in $V_2$ is non-decreasing.
\begin{claim}\label{cl:util-nodec}
For all $i \in V_2(t)$, $u_i(\ff(t+1)) \geq u_i(\ff(t))$.
\end{claim}
\begin{proof}
At time $t$ suppose $k \in V_2(t)$ plays a best-response. By Claim \ref{cl:util-inc}, $u_k(\ff(t+1)) > u_k(\ff(t))$. Now consider an $i\in V_2(t)$ other than $k$. The only way $i$'s utility decreases is if $i$ is a neighbor of $k$, and $k$ decreases her proposal to $i$. Since, $Sl_i(i)=0$, it must be the case that $f_{ki}(t)=f_{ik}(t)$, and after the $k$ changes we have $f_{ki}(t+1) \le f_{ik}(t)-\eta$. This implies $Sl_i(t+1)\ge \eta$. Since $i \in V_2(t)$ we already knew that $|W_i(t)|>0$. Since $k\notin W_i(t)$, at time $(t+1)$ the set $W_i$ is same as at $t$. Thus, we have $Sl_i(t+1)>0$ and $|W_i(tِ+1)|>0$. Then, by Lemma \ref{clm:personal} the total slack has to decrease after some time, a contradiction to the hypothesis that the total slack is constant. 
\end{proof}

To conclude, by Claim \ref{clm:v1non-decr} set $V_1$ is monotonically increasing, while for players in set $V_2$ the utility is non-decreasing (Claim \ref{cl:util-nodec}). Furthermore, since all players in $V_1$ are at their best-response, in every round a player in $V_2$ changes strategy and by Claim \ref{cl:util-inc} increases her utility by a non-trivial amount. Now, if the maximum utility an agent can achieve is bounded, then the dynamics has to converge in finite time and the proof follows.

Next, we calculate the maximum possible utility that $i$ could ever obtain. This
is given by the following convex program which describes the ideal frequency allocation according to player $i$, ignoring
contraints put by the other players
\begin{equation} \label{eq:localconv}
\begin{array}{lll}
max  & \qquad \displaystyle{\sum_{\mathclap{\substack{j \in \CN_i \\ j \neq i}}} {w_{ij} u_{ij}(f^*_{ij}) } } & \\
s.t. & \qquad \sum_{j \in \CN_i} f_{ij} \leq \beta_i & \\
     & \qquad f_{ij} \geq 0 & \forall j \in \CN_i
\end{array}
\end{equation}
Let $OPT_i$ be the solution to \eqref{eq:localconv}. Notice that $i$'s proposals here are not restricted by $j$'s proposals,
which implies that $i$ cannot get more utility than $OPT_i$ with budget $\beta_i$. Therefore, player $i\in V_2$ can only change
strategies to increase her utility at most $\frac{OPT_i}{\Delta_i}$ times, after which she cannot increase her utility anymore.
Thus the proof follows using Claims \ref{cl:util-inc} and \ref{cl:util-nodec}.


\section{Missing Proofs of Section \ref{sec:proper}}\label{app:proper}

\subsection{Proof of Lemma \ref{lem:match}}\label{app:match}

We construct a new frequency profile $\ff'$ from $\ff$ such that $SW(\ff') = SW(\ff)$ and $\ff'$ is a \NE, in
the following way. For every player $i \in \players$, we look at $i$'s proposals to her neighbors in $\ff$. For every
$j \in \CN_i$ such that $f_{ij} > f_{ji}$, we adjust $i$'s proposal to $j$ so that $i$ matches $j$'s proposal to her,
thus having $f'_{ij} = f_{ji}$. In other words, we make $i$ into a pessimistic player.

Since this procedure is done for all players, it follows that every player in $W_i(\ff)$ lowered their proposal to $i$
to match $i$'s proposal to them. Thus, $\forall j \in W_i(\ff)$, in $\ff'$, we have $f'_{ji} = f_{ij} = f'_{ij}$.
Since $i$ is now matching all frequency proposals of her neighbors, it follows that $|W_i(\ff')| = 0$. Since this
condition holds for every $i \in \players$, from Lemma \ref{lem:wineq} we get that $\ff'$ is a \NE.

We finally show that $SW(\ff') = SW(\ff)$. Indeed, since the only difference between $\ff'$ and $\ff$ are the
proposals of player $i$ to $j$ for the players $j$ such that $f_{ij} > f_{ji}$, we understand that in $i$'s utility
calculation, $i$ already satisfied the constraint $f^*_{ij} \leq f_{ji}$ with equality. Thus, this decrease in
proposal from $f_{ij} > f_{ji}$ to $f'_{ij} = f_{ji}$, does not affect $i$'s utility at all. Thus, for every player
$i \in \players$, we get that $u_i(\ff') = u_i(\ff)$, which implies $SW(\ff') = SW(\ff)$.

\subsection{Proof of Theorem \ref{theor:poa}}\label{app:poa}

Consider the game with players on a two-dimensional grid presented in Fig. \eqref{fig:bad-poa-grid}. The graph in this game
is $4$-regular and every player has exactly $4$ neighbors. Given a player $i$ in row $k$, she has one neighbor in row $k+1$,
two neighbors in row $k$, and one neighbor in row $k-1$. We assume that, for every player $i$ in row $k$,
$w_{ij} = \frac{1}{2} - \eps$ if $j$ is a vertical neighbor (either in row $k-1$ or in row $k+1$), while $w_{ij} = \eps$ if
$j$ is a horizontal neighbor, i.e. $j$ is in row $k$. Thus, every player wants to interact more with her vertical
neighbors and wants to have almost no interaction with her horizontal neighbors. Furthermore, we assume that all
players $i \in \players$ have the same resource constraint $\beta_i = \beta$, and that for all players $i,j$ where $j \in \CN_i$,
the utility that $i$ gets from $j$ is $u_{ij}(f^*_{ij}) = f^*_{ij}(\beta - f^*_{ij})$ for $0 \leq f^*_{ij} \leq \frac{\beta}{2}$.

\begin{figure}[!htbp]
 \centering
 \begin{subfigure}[b]{0.45\linewidth}
  \centering
  \includegraphics[width=\linewidth]{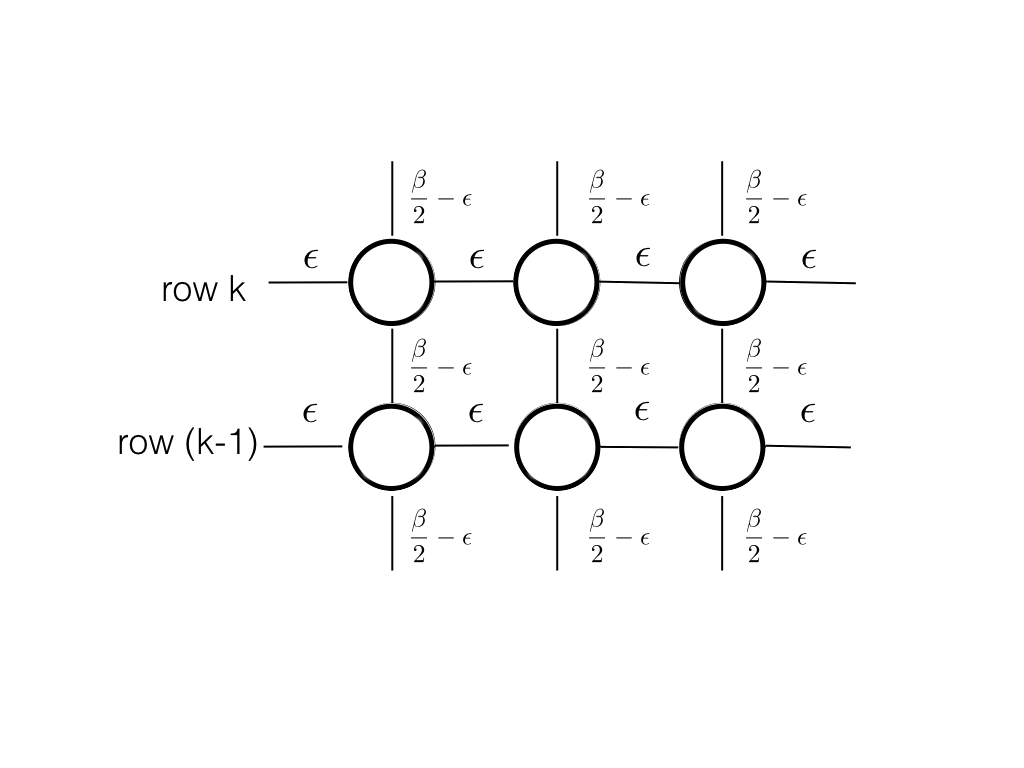}
  \caption{Frequency profile with high social welfare.}\label{fig:bad-poa-grid:good}
 \end{subfigure}
 ~
 \begin{subfigure}[b]{0.45\linewidth}
  \centering
  \includegraphics[width=\linewidth]{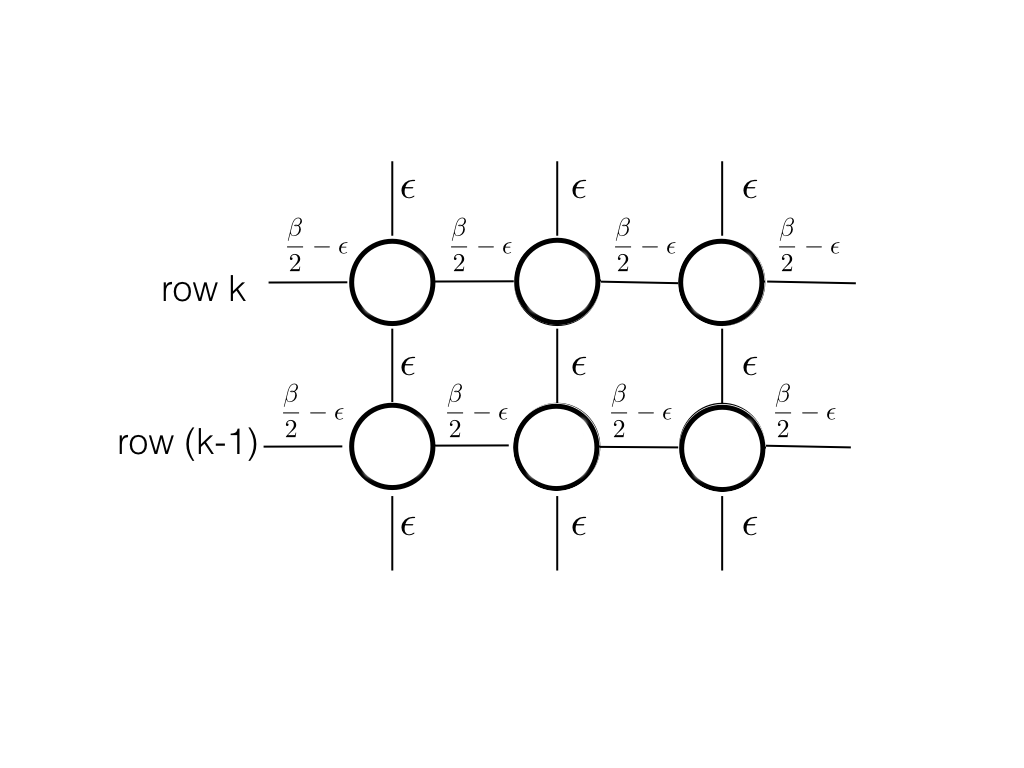}
  \caption{NE with low social welfare.}\label{fig:bad-poa-grid:bad}
 \end{subfigure}
 \caption{Example of a game with arbitrarily bad PoA.}\label{fig:bad-poa-grid}
\end{figure}

Consider two frequency profiles $\ff_{good}$ and $\ff_{bad}$. In $\ff_{good}$, we have every player $i$ assigning
frequency $\frac{\beta}{2} - \eps$ to her vertical neighbors and $\eps$ to her horizontal neighbors, as shown in
Fig. \eqref{fig:bad-poa-grid:good}. In $\ff_{bad}$, every player $i$ assigns frequency $\eps$ to her vertical neighbors
and $\frac{\beta}{2} - \eps$ to her horizontal neighbors, as shown in Fig. \eqref{fig:bad-poa-grid:bad}. We show that
$\ff_{good}$ corresponds to a frequency profile with good social welfare, while $\ff_{bad}$ corresponds to a \NE with
bad social welfare.

First of all, notice that in $\ff_{bad}$ every player is matching all her neighbors' frequency proposals. Thus, for
every player $i \in \players$, we have $|W_i(\ff_{bad})| = 0$, which implies that $\ff_{bad}$ is a \NE, by Lemma
\ref{lem:wineq}. No player has positive slack and thus, no player can gain more utility by unilaterally deviating from
her proposals.

We proceed by calculating the social welfare in $\ff_{good}$ and $\ff_{bad}$. In $\ff_{good}$, every player $i$ gets utility

\[
u_i(\ff_{good}) = 2 \left( \frac{1}{2} - \eps \right) \left( \frac{\beta}{2} - \eps \right) \left( \frac{\beta}{2} + \eps \right) + 2 \eps^2 \left( \beta - \eps \right)
\]

while in $\ff_{bad}$, every player $i$ gets utility

\[
u_i(\ff_{bad}) = 2 \eps \left( \frac{\beta}{2} - \eps \right) \left( \frac{\beta}{2} + \eps \right) + 2 \left( \frac{1}{2} - \eps \right) \eps \left( \beta - \eps \right)
\]

Therefore, the social welfare in the two profiles is

\[
SW(\ff_{good}) = n \left( 2 \left( \frac{1}{2} - \eps \right) \left( \frac{\beta}{2} - \eps \right) \left( \frac{\beta}{2} + \eps \right) + 2 \eps^2 \left( \beta - \eps \right) \right)
\]

and

\[
SW(\ff_{bad}) = n \left( 2 \eps \left( \frac{\beta}{2} - \eps \right) \left( \frac{\beta}{2} + \eps \right) + 2 \left( \frac{1}{2} - \eps \right) \eps \left( \beta - \eps \right) \right)
\]

Let $OPT$ be the solution profile that maximizes the social welfare, and let $NE_{worst}$ be the \NE that minimizes the
social welfare. We have

\begin{equation}\label{eq:poa}
PoA = \frac{SW(OPT)}{SW(NE_{worst})} \geq \frac{SW(\ff_{good})}{SW(\ff_{bad})} = \frac{n \left( 2 \left( \frac{1}{2} - \eps \right) \left( \frac{\beta}{2} - \eps \right) \left( \frac{\beta}{2} + \eps \right) + 2 \eps^2 \left( \beta - \eps \right) \right)}{n \left( 2 \eps \left( \frac{\beta}{2} - \eps \right) \left( \frac{\beta}{2} + \eps \right) + 2 \left( \frac{1}{2} - \eps \right) \eps \left( \beta - \eps \right) \right)}
\end{equation}

It is easy to see that $\displaystyle{\lim_{\eps \to 0} {PoA} = + \infty}$. Thus, we can make the PoA arbitrarily large
by decreasing $\eps$. Notice that this result does not depend on the size of the grid, so the game can have bad \NEa even
with a small number of players.

\subsection{Proof of Lemma \ref{lem:opt-pes}}\label{app:opt-pes}

The first part of the statement follows immediately from Lemma \ref{lem:match}. We want to show that every 
convex combination of $\ff$ and $\ff'$ is also an optimistic \NE. Let $\alpha \in (0, 1]$, and
$\ff'' = \alpha \ff + (1 - \alpha) \ff'$. If $\alpha = 0$, then $\ff'' = \ff'$ and it is a pessimistic \NE, thus we focus
on all other values of $\alpha$. Consider a pair of players $i, j \in \players$ such that $f_{ij} > f_{ji}$ in $\ff$.
Obviously, $f'_{ij} = f'_{ji}$ in $\ff'$.

By our construction of $\ff'$ from Lemma \ref{lem:match}, we have that $f_{ji} = f'_{ji}$. Thus, we understand that
$f''_{ji} = f'_{ji} = f_{ji}$ in $\ff''$. For $i$'s frequency proposal we have

\begin{equation}\label{eq:opt-pes}
f''_{ij} = \alpha f_{ij} + (1 - \alpha) f'_{ij}
\end{equation}

Since $f_{ij} > f_{ji}$ in $\ff$, let $f_{ij} = f_{ji} + \delta$, where $\delta > 0$. Then, \eqref{eq:opt-pes} becomes

\begin{align*}
f''_{ij} & = \alpha (f_{ji} + \delta) + (1 - \alpha) f'_{ij} = \alpha (f_{ji} + \delta) + (1 - \alpha) f'_{ji} = \alpha (f_{ji} + \delta) + (1 - \alpha) f_{ji} \\
& = f_{ji} + \alpha \cdot \delta
\end{align*}

We see that $i$'s proposal to $j$ in $\ff''$ is $f''_{ij} = f_{ji} + \alpha \cdot \delta > f_{ji}$. Since $\ff$ is a \NE,
player $j$ gained no utility by deviating from their strategy and increasing their proposal to $i$. In $\ff''$ $i$ still
proposes higher than $j$ in their interaction, but her proposal is a lower one overall. Because $j$ did not gain any utility
by deviating in $\ff$, she will still not gain any utility by deviating in $\ff''$. Since this holds for every such pair
$i, j \in \players$, we understand that $\ff''$ is an optimistic \NE.

\section{Characterization of the Best-Response}\label{app:br}

In this section, we attempt to understand the players' best-response in depth. As stated in section \ref{sec:prelim},
at each turn $t$, the player that is picked to update her proposals solves the local convex program \eqref{eq:BR} in order
to maximize her utility. We look at the KKT conditions of \eqref{eq:BR} in order to obtain better intuition as to how this
update is performed.

Let $i$ be the player that updates her opinion at time $t$, for a given profile $\ff_{-i}$. We look at the dual of
\eqref{eq:BR}. Let $\lambda_{ij}$ be the dual variable corresponding to the constraint $f^*_{ij} \leq f_{ij}$,
$\lambda_{ji}$ be the dual variable corresponding to $f^*_{ij} \leq f^C_{ji}$, and $\delta_i$ be the dual variable
corresponding to $\sum_{j \in \CN_i} {f_{ij}} \leq \beta_i$. The KKT conditions of \eqref{eq:BR} are the following

\begin{align}\label{eq:br-kkt}
\begin{array}{rlr}
w_{ij} \frac{\prt u_{ij}}{\prt f^*_{ij}} \leq \lambda_{ij} + \lambda_{ji} & \bot \quad f^*_{ij} \geq 0 & \qquad \forall j \in \CN_i \\
\lambda_{ij} \leq \delta_i & \bot \quad f_{ij} \geq 0 & \qquad \forall j \in \CN_i
\end{array}
\end{align}

where $\bot$ denotes the complementarity of the conditions, meaning that one of the two inequalities on the left and on
the right side of $\bot$ have to be tight, i.e. hold with equality. Let us reorder the players in $\CN_i$ in the following
order, where

\[
w_{i1} \frac{\prt u_{i1}}{\prt f^*_{i1}}\bigg|_{f^*_{i1} = 0} \geq w_{i2} \frac{\prt u_{i2}}{\prt f^*_{i2}}\bigg|_{f^*_{i2} = 0} \geq \cdots \geq w_{ik} \frac{\prt u_{ik}}{\prt f^*_{ik}}\bigg|_{f^*_{ik} = 0}
\]

where $k = |\CN_i|$.

Initially, let all $f_{ij} = 0$. Now $i$ tries to maximize her utility, by increasing $f_{i1}$, because
$w_{i1} \frac{\prt u_{i1}}{\prt f^*_{i1}}\Big|_{f^*_{i1} = 0}$ is the highest in her neighbourhood. $i$ will keep
increasing $f_{i1}$ alone, until
$w_{i1} \frac{\prt u_{i1}}{\prt f^*_{ij}} = w_{i2} \frac{\prt u_{i2}}{\prt f^*_{i2}}\Big|_{f^*_{i2} = 0}$. At that point,
$i$ will get the same increase in utility by both players $1$ and $2$. Thus, $i$ will start increasing both $f_{i1}$ and
$f_{i2}$ simultaneously. Again, this will be optimal until
$w_{i1} \frac{\prt u_{i1}}{\prt f^*_{i1}} = w_{i2} \frac{\prt u_{i2}}{\prt f^*_{i2}} = w_{i3} \frac{\prt u_{i3}}{\prt f^*_{i3}}\Big|_{f^*_{i3} = 0}$.
This process continues in a similar manner.

If during the process we get $f_{ij} = f_{ji}$ for some $j \in \CN_i$ as $f_{ij}$ increases, then $i$ will stop increasing
$f_{ij}$ because $f^*_{ij} = f_{ji}$, and $i$'s utility will not increase by further increasing $f_{ij}$, as indicated by
\eqref{eq:BR}. Instead $i$ will start increasing $\lambda_{ij}$ to accommodate for that fact. In the end, we either have
$\sum_{j \in \CN_i} {f_{ij}} = \beta_i$ and $i$ has allocated fully her resource, or
$f_{ij} = f_{ji} \quad \forall j \in \CN_i$, and thus $Sl_i > 0$.

It is clear that via this analysis we get a full characterization of $i$'s best-response. This iterative process helps us
understand our game more as well as provides some intuition as to why best-response dynamics converge. Finally, this
analysis help us model the best-response in our experiments, as discussed in section \ref{sec:exper}.